\documentclass[12pt]{article}

\usepackage[colorlinks=true,linktoc=page]{hyperref}

\usepackage{url}

\usepackage{graphicx,psfrag}
\usepackage{caption}
\usepackage{amsmath,amsthm,amsfonts,amssymb}
\usepackage{enumitem}
\usepackage{comment}

\usepackage{authblk}

\newtheorem{thm}{Theorem}

\newtheorem{prop}{Proposition}
\newtheorem{defn}{Definition}

\newtheorem{lemma}{Lemma}

\newtheorem{remark}{Remark}

\numberwithin{equation}{section}

\usepackage{xcolor}

\newcommand {\Om}{\Omega}

\usepackage{xcolor}
\usepackage{graphicx}
\usepackage{import} 

\definecolor{inkscapeblue}{HTML}{476A9A}
\definecolor{inkscapered}{HTML}{ec1818}
\definecolor{inkscapeorange}{HTML}{bda343}
\definecolor{inkscapegreen}{HTML}{2d5124}
\definecolor{inkscapegray}{HTML}{000000}

\title{Spacetime Bartnik Mass Positivity and Temporal Monotonicity for Black Holes}

\author[1]{Lars Andersson}
\affil[1]{Kungliga Tekniska högskolan\\~~~~~~~~~~~~~~~~~~~~~~~~~~~~~~~~~~~~~~~~~~~~~~~ 100 44 Stockholm, 
Sweden\\~~~~~~~~~~~~~~~~~~~~~~~~~~~~~~~~~~~~~~~~~~~~~~~~~~~~~~~~~~
larand@kth.se}

\author[2]{Marcus Khuri} 
\affil[2]{Department of Mathematics,
Stony Brook University\\~~~~~~~~~~~~~~~
Stony Brook, NY 11794, USA\\~~~~~~~~~~~~~~~~~
marcus.khuri@stonybrook.edu}

\author[3]{Marc Mars} 
\affil[3]{Institute of Fundamental Physics and Mathematics\\
Universidad de Salamanca\\ Plaza de la Merced s/n\\~~~~~~~~~~~~~~~~~~ E-37008 Salamanca, Spain\\~~~~~~~~~~~~~~~~~~~~~~~~~~~~~~~~~~~~~~~~~~~~~~~~~~~~~~~~~~~~~~~~
marc@usal.es}

\author[4]{Walter Simon} 
\affil[4]{Fakult\"at f\"ur Physik, Universit\"at Wien\\~~~~~~~~~~~~~~
Boltzmanngasse 5, 1090 Wien, Austria\\~~~~~~~~~~~~~~~
walter.simon@univie.ac.at}

\date{}

\begin{document}

\maketitle

\begin{abstract}
We define a quasilocal mass of Bartnik type, and establish its positivity and temporal monotonicity
properties for two classes of domains associated with black holes. More precisely, we first show that the quasilocal
mass is strictly positive for spacelike hypersurfaces that are: compact with apparent horizon boundary or noncompact with asymptotically flat ends and containing an apparent horizon in any admissible extension. Secondly, we show that the quasilocal mass is monotonically nondecreasing in time within evolutionary scenarios related to the two aforementioned settings.
\end{abstract}

\section{Introduction}
\label{sec1} \setcounter{equation}{0}
\setcounter{section}{1}

\subsection{Motivation} 

Gravitational energy is non-local  and hence it is an interesting problem to find useful quasilocal notions of mass or energy-momentum \cite{Pen82}. Given a spacelike hypersurface
$\Omega$ in spacetime, its quasilocal mass should measure the total mass or energy-momentum contained in $\Omega$. There are many competing definitions of quasilocal masses with interesting properties \cite{Sza09}.

Here we consider a quasilocal mass $m(\Om)$ of Bartnik type 
\cite{RB1,RB2,RB3,SM}. We recall that Bartnik's idea is to define `admissible' extensions of $\Omega$ which are 
asymptotically flat, and satisfy an energy condition such that the ADM masses of the extensions are well defined and nonnegative. A key issue is to find a suitable `no-horizons' condition, that should render the infimum of the ADM masses at a designated end of all admissible extensions nonnegative, and positive unless the original domain $\Omega$ arises from Minkowski space. This problem can be studied in the Riemannian setting, where the no-horizons condition is expressed in terms of the absence of minimal surfaces, or in the general case of initial data sets where one might expect apparent horizons to play a role.
While the Riemannian case has been studied extensively and is fairly well understood, there is very little known about the general setting involving initial data. This discrepancy is in part due to the different status of the
Penrose inequality in the respective settings.

In this work we are on the one hand concerned with positivity proofs for $m(\Om)$, (cf Theorems \ref{comp} and \ref{noncomp}), as applications of the recently obtained Penrose-like inequality  \cite[Theorem 1.1]{ABKK}. On the other hand, in a spacetime foliated by spacelike hypersurfaces, we prove short-time monotonicity 
of $m(\Om)$ along marginally outer trapped tubes as well as for generic achronal
surfaces. The former setting (Theorem \ref{moncomp}) fits the intuitive picture of a ‘black hole’ which swallows energy upon time evolution. It may be compared with, although distinguished from,
the well-known area laws, which state that the areas of event horizons [20,
15] and dynamical horizons [9], do not decrease to the future. 
On the other hand, our second monotonicity result (Theorem \ref{monnoncomp}) seems much less intuitive, as the achronal boundary is not specified. However, the topologcal requirements for this result are such that apparent horizons are necessarily present as well, if only in any admissible extension.

We anticipate here that in the positivity Theorems \ref{comp} and \ref{noncomp} we admit data which satisfy the dominant energy condition, while we restrict ourselves to vacuum domains in Theorems
\ref{moncomp} and \ref{monnoncomp}. Nevertheless, for technical reasons we need to admit matter fields 
in the admissible extensions of the monotonicity results as well.

\subsection{Standard Definitions}
\label{std}
Throughout the paper all manifolds are assumed to be smooth, connected 
and orientable unless indicated otherwise. An \textit{initial data set} for the Einstein equations is a triple $(M,g,k)$, consisting of a $3$-dimensional manifold $M$ (possibly with boundary), a complete Riemannian metric $g$, along with a symmetric 2-tensor $k$ representing the second fundamental form of an embedding into spacetime.  These quantities are assumed to be smooth and satisfy the constraint equations
\begin{equation}
16 \pi \mu=R+(\text{Tr}_g k)^2-|k|^2, \qquad\text{ }\text{ }
8\pi J=\mathrm{div}\left(k-(\mathrm{Tr}_g k)g\right),
\end{equation}
where $R$ is the scalar curvature of $g$, and $\mu$, $J$ are the energy-momentum density of matter fields. 
The \emph{dominant energy condition} holds if $\mu\ge |J|$. Moreover, the data will be referred to as \emph{asymptotically flat} (AF) if there exists a compact subset $\mathcal{K}$ for which $M\setminus \mathcal{K}=\cup_{a=1}^{a_0}M^{a}_{end}$, so that the ends $M_{end}^{a}$ are pairwise disjoint and each is diffeomorphic to the complement of a Euclidean ball $\mathbb{R}^3 \setminus B$. Furthermore, if $\varphi$ is the diffeomorphism from Euclidean space with Cartesian coordinates $x$ to an end then
\begin{align}\label{1}
\begin{split}
\varphi^* g=\delta+O_2(|x|^{-q}), \qquad\quad
\varphi^* k= O_1(|x|^{-q-1}),\\
\varphi^*\mu, \varphi^*J=O(|x|^{-2q-2}),\qquad\quad \varphi^*\mathrm{Tr}_g k=O(|x|^{-2q-1}),
\end{split}
\end{align}
for some $q>\tfrac{1}{2}$ where $O_l(|x|^{-q})$ represents a tensor in the weighted space $C^l_{-q}(\mathbb{R}^3)$. Note that the additional decay on the trace of $k$ is usually not included in the AF definition, however it is included here to facilitate the use of \cite{ABKK}. The ADM energy and linear momentum of each end are well-defined \cite{Bartnik,Chrusciel} with these asymptotics and are given by
\begin{equation}
E=\lim_{r\rightarrow\infty}\frac{1}{16\pi}
\int_{S_{r}}\sum_i (g_{ij,i}-g_{ii,j})\nu^{j}dA,
\end{equation}
\begin{equation}
P_i =\lim_{r\rightarrow\infty}\frac{1}{8\pi}
\int_{S_{r}}(k_{ij}-(\text{Tr}_g k)g_{ij})\nu^{j}dA,
\end{equation}
where $S_r$ are coordinate spheres with unit outer normal $\nu$ and area element $dA$.
The ADM mass is then the Lorentz length of the energy-momentum vector, that is $m=\sqrt{E^2-|P|^2}$.

Consider a closed separating hypersurface $\Sigma\subset M$ with null expansions $\theta_{\pm}= H \pm \mathrm{Tr}_{\Sigma} k$. Here, $H$ denotes the mean curvature of $\Sigma$ obtained as the tangential divergence of the unit normal $\nu$ pointing towards a designated end $M_{end}^1$. The null expansions are themselves (spacetime) mean curvatures, namely in the null directions $\nu\pm n$ where $n$ represents the future pointing timelike unit normal to the slice $(M,g,k)$. These quantities can be interpreted physically as determining the rate of change of area for a shell of light emanating from the surface in the outward future/past direction, and thus may be used to assess the strength of the gravitational field. The gravitational field is interpreted as strong near the surface $\Sigma$ if it is \textit{outer or inner trapped}, that is $\theta_+< 0$ or $\theta_- <0$. Moreover, $\Sigma$ is called a \emph{marginally outer or inner trapped surface} (MOTS or MITS) when $\theta_+ = 0$ or $\theta_-=0$. These types of surfaces are also referred to as future or past apparent horizons, and naturally arise as boundaries of future or past trapped regions \cite{AM}. Moreover, a collection $\Sigma$ of disjoint MOTS and MITS components will be called an \textit{outermost apparent horizon} with respect to $M_{end}^1$, if $\Sigma$ is not enclosed from the perspective of $M_{end}^1$ by any other disjoint collection of apparent horizon components. The existence of an outermost apparent horizon for each end follows from \cite{AM}, by first finding the outermost MOTS and outermost MITS separately, and then removing components until all are disjoint. Although the outermost MOTS and outermost MITS are individually unique, the second step in which components are removed entails a choice, and thus the resulting outermost apparent horizon may not be unique.

\subsection{ Definitions of Bartnik type}

\label{Bld}

\begin{defn}[\bf $\mathbf{b}$-Admissible Extension]\label{ae}
Let $(\Omega,\mathrm{g}, \mathrm{k})$ be an initial data set with nonempty boundary $\partial\Omega$, such that $\Omega$ is either compact or possesses $a_0 >0$ AF ends. Let $\mathbf{b}$ denote a collection of boundary components. A $\mathbf{b}$-admissible extension of $(\Omega,\mathrm{g},\mathrm{k})$ 
is an AF initial data set $(M,g,k)$ without boundary satisfying the following conditions.

\begin{enumerate}
\item $(M,g,k)$ has one end when $\Omega$ is compact and has $a_0 +1$ ends otherwise. The additional end will be
referred to as the designated end, and will be denoted by $M_{end}^1$.

\item The boundary of the unbounded component of $M\setminus\Omega$ consists of the $\mathbf{b}$ boundary components of $\Omega$.

\item $(\Omega, \mathrm{g}, \mathrm{k})$ arises from an isometric embedding into $(M,g,k)$ which preserves the extrinsic curvature.

\item $(M,g,k)$ satisfies the dominant energy condition.

\item Any closed minimal hypersurface of $M$ is contained within $\Om$.

\end{enumerate}

\end{defn}

\begin{figure}[h!]
    \centering
\includegraphics[width=1.1\textwidth]{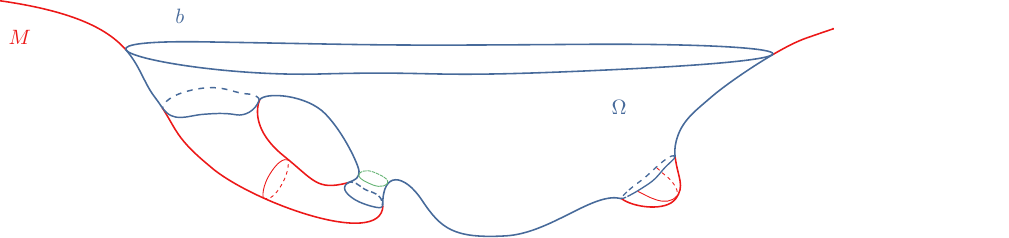}
    \caption{Example of an admissible extension as in Definition \ref{ae}. The domain $\Omega$ is depicted in blue while the extension $M \setminus \Omega$ is shown in red. In this case, the collection of boundary components 
    $\mathbf{b}$ has a single element. 
      Note that the extended manifold $(M,g)$ has a closed minimal surface (drawn in green) contained in $\Omega$.}
      \label{Fig1}
  \end{figure}
  
\begin{remark}
For technical reasons concerning the proofs of the temporal monotonicity Theorems \ref{moncomp} and \ref{monnoncomp}, it will be convenient to also introduce the following slightly more stringent condition in place of (5).
\begin{enumerate}
\item[5*.] Any closed minimal hypersurface of $M$ is contained within the interior int($\Om$) of $\Om$.
\end{enumerate}
In Definition \ref{ae}, if condition (5) is replaced by (5*), we shall refer to the extension as \emph{strictly $\mathbf{b}$-admissible}. Note that while condition (5) admits minimal surfaces as components of $\partial\Omega$, this is no longer the case for (5*).  

\end{remark}

\begin{remark}
When $\Omega$ has several boundary components, the subcollection  determined by $\mathbf{b}$ will form the boundary of the noncompact component of $M\setminus\Omega$ containing the designated end, while the remaining components will be `capped-off'  by the compact components of $M\setminus\Omega$. 
Here `capping off' includes the possibility that subsets of such remaining boundary components are joined by `tubes', as long as the resulting  construction respects conditions (5) or (5*). See Fig. \ref{Fig1} for an example.
\end{remark}

\begin{remark}
 \label{by} We avoid  defining extensions $(M,g,k)$ with boundaries; the reason is  technical and 
 arises from the proof of the temporal monotonicity results Theorems \ref{moncomp} and \ref{monnoncomp}. 
\end{remark}

\begin{remark} 
\label{nomin}  
We recall that in Bartnik's original definition of admissible extensions in the Riemannian setting \cite{RB1}, minimal surfaces are not admitted anywhere (not even inside $\Om$). A corresponding modification of point (5) of the above definition is viable only if $\Om$ is compact and the extensions have trivial second homology.
\end{remark}
\begin{remark}
As to the definition by Huisken and Ilmanen in the Riemannian context 
\cite[Section 9]{HI}, it relates to the above one with the following differences.
\begin{itemize}
 \item  In our notion of admissibility, minimal surfaces are allowed within $\Omega$ which is a closed subset of the extension, while in \cite{HI} minimal surfaces are only allowed to occur within the boundary of the extension and the extension contains $\Omega$ as a bounded open subset.
 \item As mentioned in Remark \ref{by} we do not admit extensions with boundaries, in contrast to \cite{HI}.
\end{itemize}
\end{remark}

\begin{defn}[\bf Bartnik-Type Mass] \label{bm}
Let $(\Omega,\mathrm{g}, \mathrm{k})$ be an initial data set with boundary, such that $\Omega$ is either compact or possesses a finite number of AF ends. If the set of $\mathbf{b}$-admissible extensions is nonempty, then a Bartnik mass\footnote{If $\mathbf{b}$ consists of all boundary components then the qualifier $\mathbf{b}$ will be removed from the notation.} of this data is defined to be 
\begin{equation*}
\label{mgen}
m_{\mathbf{b}}(\Omega) =  \inf \{m_{ADM}(M_{end}^1,g,k)\mid (M,g,k) ~\mbox{is a $\mathbf{b}$-admissible extension of} ~\Om \}, 
\end{equation*}
where $m_{ADM}(M_{end}^1,g,k)$ is the ADM mass of the designated end.
\end{defn}
We remark that the domain $\Om$ is restricted by the requirement that there exists an admissible extension. (Note that we do not admit infinite values for $m_{\mathbf b}$). In particular, this excludes domains with mean concave ($H \le 0$) boundary. We anticipate that our positivity results Theorems \ref{comp} and \ref{noncomp} do not make any further requirement on $\partial\Om$,
while the monotonicity results Theorems \ref{moncomp} and \ref{monnoncomp} require strict mean convexity ($H > 0$).

\subsection{Positivity}

\label{posi} 

The positive mass theorem \cite{SY,Witten} implies that the Bartnik mass is always nonnegative, and the question of its strict positivity outside of the ground state was discussed in Bartnik's original proposal \cite{RB1,RB2}. This has been satisfactorily answered in the time symmetric case by Huisken-Ilmanen \cite[Positivity Property 9.1]{HI} and Dong-Song \cite[Theorem 5.1]{DS}, see also the discussion in Anderson-Jauregui \cite{AJ}. However, it does not appear that this question has been previously addressed in the spacetime case. Here we will establish two different strict positivity results, both of which are associated with the presence of an apparent horizon.

\begin{thm}[\bf Positivity for Compact Data]
\label{comp}
Let $(\Omega,\mathrm{g}, \mathrm{k})$ be an initial data set with boundary, 
such that $\Omega$ is compact. Assume that a nonempty subset of boundary components $\mathbf{b}$
consists entirely of MOTS and MITS, and that a $\mathbf{b}$-admissible
extension exists. Then the corresponding Bartnik mass is strictly positive,
$m_{\mathbf{b}}(\Omega)>0$.
\end{thm}

\begin{thm}[\bf Positivity for Data With AF Ends]
\label{noncomp}
Let $(\Omega,\mathrm{g}, \mathrm{k})$ be an initial data set with boundary, 
such that $\Omega$ possesses a finite nonzero number of AF ends. 
Let $\mathbf{b}$ denote a collection of boundary components.
If a $\mathbf{b}$-admissible extension exists, then the corresponding Bartnik 
mass is strictly positive, $m_{\mathbf{b}}(\Omega)>0$.
\end{thm}

Both theorems may be considered as `strong field' results in the sense that a MOTS/MITS
is always present, either in the boundary of $\Om$ as is the case for Theorem \ref{comp}, or 
in any admissible extension of $\Omega$ as is the case for Theorem \ref{noncomp} according to the barrier approach for MOTS existence  given by \cite[Theorem 1.1]{AM} and \cite[Theorem 1.1]{Eichmair}.

\subsection{Monotonicity} 

\label{mono}

We now turn our attention to temporal monotonicity properties of the Bartnik mass, a topic which does not appear to have been previously addressed in the literature. 
We first recall the monotonicity problem for nested domains
$\Om_1 \subset \Om_2 \subset M$ within a Riemannian manifold $M$. This holds trivially (in the sense that $m(\Omega_1) \le m(\Omega_2)$) for Bartnik's original definition \cite{RB1} and hence for the corresponding modification of 
Definition \ref{ae}, cf. Remark \ref{nomin}.  However, monotonicity is no longer guaranteed for Definition \ref{ae} (or any other definitions which admit minimal surfaces inside $\Om$). In this context we recall the  monotonicity result of the general (spacetime) Bartnik mass \cite{MM}, still for nested surfaces within a smooth initial data set. However, this result supposes the existence of stationary minimizing extensions, which are not assumed in the present work. Moreover, such extensions are unlikely to exist for generic, stable MOTS, by analogy with the corresponding result for stable minimal surfaces \cite{MS}.

In contrast, the next pair of results applies to foliated spacetimes, and requires some preparatory discussion. Consider a 4-dimensional spacetime $\mathcal{M}$ foliated by spacelike hypersurfaces $\{N_t\}_{t\in I}$, where $I\subset\mathbb{R}$ is a nonempty interval. A \textit{marginally outer trapped tube} (MOTT) (adapted to this foliation) is a hypersurface $\mathcal{H}\subset\mathcal{M}$ which is foliated by MOTS $\Sigma_t =\mathcal{H}\cap N_t$. A corresponding \textit{MOTT domain} is a 4-dimensional submanifold of $\mathcal{M}$ with boundary that is foliated by spacelike hypersurfaces $\Omega_t \subset N_t$, such that $\partial\Omega_t =\Sigma_t$.

\begin{figure}[h!]
    \centering
    \includegraphics[width=1.0\textwidth]{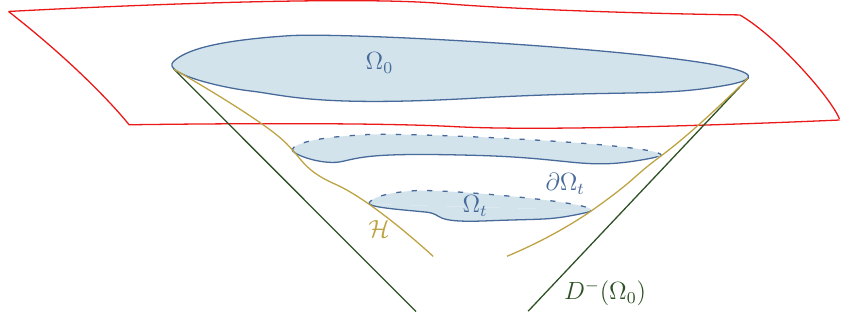}
    \caption{Visualization of the setting of the Monotonicity Theorems \ref{moncomp} and \ref{monnoncomp}. In Theorem \ref{moncomp}, the hypersurface 
    ${\mathcal H}$ is the smooth MOTT defined by the past evolution of $\partial
        \Omega_0$ with respect to the given foliation $(\Omega_t, \mathrm{g}_t, \mathrm{k}_t)_{t \in [- \varepsilon, 0]}$. In Theorem 4 ${\mathcal H}$  is any smooth spacelike hypersurface contained in  $D^{-} (\Omega_0)$ and containing $\partial \Omega_0$.}
          \label{Fig2}
        \end{figure}

  If a spacetime with foliation has the property that some leaf, say $N_0$, possesses a strictly stable MOTS $\Sigma_0 =\partial\Omega_0$, then the short time 
existence of a smooth MOTT and corresponding MOTT domain are given by Andersson-Mars-Simon \cite[Theorem 1]{AMS1}. Here, strict stability refers to the property that the principal eigenvalue of the MOTS stability operator \cite[Definition 3.1]{AMS2} is strictly positive. If $\mathcal{M}$ satisfies the null energy condition then the MOTT is achronal near $\Sigma_0$, and if further the null second fundamental form of $\Sigma_0$ with respect to the outward future direction does not vanish identically then the MOTT is spacelike everywhere near $\Sigma_0$ \cite[Theorem 9.9]{AMS2}. In this situation we will also refer to the associated MOTT domain as spacelike.

\begin{thm}[\bf Monotonicity Along a MOTT] 
\label{moncomp}
Let $(\Omega_0,\mathrm{g}_0, \mathrm{k}_0)$ be a vacuum initial data set with boundary, 
such that $\Omega_0$ is compact and each boundary component is a strictly stable and strictly mean convex MOTS.
For $\varepsilon>0$ let $\{(\Omega_t,\mathrm{g}_t,\mathrm{k}_t)\}_{t\in[-\varepsilon,0]}$ be a foliation by compact spacelike hypersurfaces of a spacelike MOTT domain contained within the past domain of dependence $D^-(\Om_0)$. 

If the initial data has an admissible extension, then for each sufficiently small $t\leq 0$ there exists an admissible
extension of $(\Omega_t,\mathrm{g}_t,\mathrm{k}_t)$ and the Bartnik mass is monotonically nondecreasing, that is
$m(\Omega_{t_1})\leq m(\Omega_{t_2})$ for all $t_1\leq t_2$ sufficiently small.

\end{thm}
We note that the condition of strict mean convexity implies that the admissible
extensions mentioned above are automatically strictly admissible, by virtue of the maximum principle. This is why we only need to define a single Bartnik-type mass, namely  (\ref{mgen}), using admissible extensions.
It should also be noted that positivity of all Bartnik masses appearing in the monotonicity results is guaranteed by Theorems \ref{comp} and \ref{noncomp}.  

As mentioned in Sect. 1.1 already, Theorem \ref{moncomp} fits the intuitive picture of a `black hole' which swallows energy upon time evolution. In contrast,  the achronal boundary is not specified in the following theorem  which, however, is restricted to non-compact data. 
Of course this theorem holds in particular for MOTTs.

\begin{thm}[\bf Monotonicity Along an Achronal Tube]
\label{monnoncomp}
Let $(\Omega_0,\mathrm{g}_0, \mathrm{k}_0)$ be a vacuum initial data set with strictly mean convex boundary, 
such that $\Omega_0$ possesses a finite nonzero number of asymptotically flat ends. 
For $\varepsilon>0$ let $\{(\Omega_t,\mathrm{g}_t,\mathrm{k}_t)\}_{t\in[-\varepsilon,0]}$  be a foliation by spacelike hypersurfaces 
of a 4-dimensional submanifold with boundary within the past domain of dependence $D^-(\Om_0)$, such that each $\Omega_t$ contains corresponding asymptotically flat ends and
$\{\partial\Omega_t\}_{t\in[-\varepsilon,0]}$ forms a spacelike hypersurface. 

Then the conclusions of Theorem \ref{moncomp} hold: If the initial data
has an  admissible extension, then for each sufficiently small $t\leq 0$ there exists an admissible extension of $(\Omega_t,\mathrm{g}_t,\mathrm{k}_t)$ and the Bartnik mass is monotonically nondecreasing, that is
$m(\Omega_{t_1})\leq m(\Omega_{t_2})$ for all $t_1\leq t_2$ sufficiently small.
\end{thm}

We note that, compared to the positivity results Theorems \ref{comp} and \ref{noncomp}, the requirements of Theorems \ref{moncomp} and \ref{monnoncomp} are more restrictive in three respects:
they hold for vacuum only, and $\partial\Om$ is required to be strictly mean convex and strictly stable (Theorem \ref{moncomp} only).

\section{Positivity Proofs}
\label{pp}

\subsection{Proof of Theorem \ref{comp}}

Let $(M,g,k)$ be a $\mathbf{b}$-admissible extension (cf Def. \ref{ae}) of the compact initial data set $(\Omega,\mathrm{g}, \mathrm{k})$ whose $\mathbf{b}$-boundary components
consist entirely of MOTS and MITS.  We recall from Sect. \ref{std} the definition and properties of an \textit{outermost apparent horizon} $\Sigma$ with respect to $M_{end}^1$.

Since $\mathbf{b}$ is nonempty the outermost apparent horizon is also nonempty, and moreover it separates $\Omega$ from infinity. We may then apply the Penrose inequality with suboptimal constant \cite[Theorem 1.1]{ABKK} to find
\begin{equation}
m_{ADM}(M_{end}^1,g,k)\geq \sqrt{\frac{\mathcal{A}}{\mathcal{C}}},
\end{equation}
where $\mathcal{A}$ is the minimum area required to enclose $\Sigma$ and $\mathcal{C}$ is a positive universal constant.
To prove strict positivity of the Bartnik mass, it then suffices to show that $\mathcal{A}$ has a uniform positive lower bound independent of the extension.

To verify the lower bound for $\mathcal{A}$, we will employ a strategy inspired by the proof of \cite[Positivity Property 9.1]{HI}. Consider a geodesic ball $B_{r}(p)\subset\Om$ centered at a point $p$ of radius $r>0$. Take $r<\mathrm{inj}(p)$ sufficiently small so that $B_{3r}(p)$ is still contained within $\Om$, and $\partial B_{s}(p)$ has positive mean curvature for all $s\in(0,2r]$. By \cite[Theorem 1.3 (iii)]{HI}, there exists a $C^{1,1}$-hypersurface $\mathcal{S}\subset M$ that represents the outermost minimal area enclosure of $\partial B_{r}(p)$ in $M$; moreover, this surface is $C^{\infty}$ and minimal where it does not contact the obstacle $B_{r}(p)$. By the maximum principle for minimal surfaces, $\mathcal{S}$ cannot lie completely within $B_{2r}(p)$ unless it agrees with $B_{r}(p)$. Furthermore, the monotonicity formula for minimal surfaces 
\cite[Theorem 17.6]{Simon} implies that any properly embedded minimal surface within $B_{3r}(p)\setminus B_{r}(p)$ that intersects $\partial B_{2r}(p)$ must have a uniform area lower bound. The last possibility is that $\mathcal{S}$ lies entirely outside of $B_{2r}(p)$. In this case $\mathcal{S}$ is a closed minimal surface, however since there are no closed minimal surfaces of $M$ that intersect $M\setminus\Omega$, we find that $\mathcal{S}\subset\Omega$. 

Next extend $(\Omega,g)$ smoothly across its boundary to a slightly larger Riemannian manifold $(\Omega',g')$ (unrelated to the Bartnik extension) in which $(\Omega,g)$ embeds isometrically;
because $\Omega$ is compact the geometry of $(\Omega',g')$ may be uniformly controlled in terms of that of $(\Omega,g)$. In particular, there exists a radius $r_0>0$ depending only on $(\Omega',g')$, such that the monotonicity formula can be applied on geodesic balls of radius $r_0$ centered at points of $\mathcal{S}$, yielding again a uniform area lower bound.
Therefore, having considered all possible scenarios for the minimal area enclosure, we find that $|\mathcal{S}|\geq c>0$ for some constant $c$ depending only on $({\Om},g)$, and since any surface that encloses $\Sigma$ must also enclose $\partial B_{r}(p)$ we have $\mathcal{A}\geq|\mathcal{S}|\geq c$.

\subsection{Proof of Theorem  \ref{noncomp}} Let $(M,g,k)$ be a $\mathbf{b}$-admissible extension of the  initial data set $(\Omega,\mathrm{g}, \mathrm{k})$, and recall that in this setting $\Omega$ possesses a finite number of asymptotically flat ends. Since large coordinate spheres in these ends are trapped from the perspective of the designated end $M_{end}^1$, and large coordinate spheres in the designated end are untrapped, it follows \cite[Theorem 3.3]{AEM} that there exists a MOTS and a MITS homologous to coordinate spheres of $M_{end}^1$. As discussed above, one may then find an outermost apparent horizon $\Sigma$ with respect to this designated end. At this stage, the same type of arguments given in the proof of Theorem \ref{comp} apply. More precisely, using the Penrose inequality with suboptimal constant it is sufficient to establish a uniform positive lower bound for $\mathcal{A}$ independent of the extension, where again $\mathcal{A}$ is the minimal area required to enclose $\Sigma$ from the perspective of $M_{end}^1$. This, in turn, may be achieved by estimating the outermost minimal area enclosure $\mathcal{S}$ (with respect to $M_{end}^1$) of the surface $\mathbf{S}_r$ formed by the union of coordinate spheres having sufficiently large radius $r$ in the asymptotically flat ends of $\Omega$. Observe that $\mathbf{S}_r$ is enclosed by $\Sigma$ with respect to $M_{end}^1$, since the maximum principle for MOTS/MITS \cite[Proposition 2.4]{AM} prevents $\Sigma$ from entering these ends. Moreover, $\mathbf{S}_r$ has negative mean curvature with respect to the inward pointing normal, and thus $\mathcal{S}$ cannot intersect $\mathbf{S}_r$ and must be a  minimal surface. By (5) of Definition \ref{ae}, it follows that $\mathcal{S}$ lies within the precompact component of $\Omega\setminus \mathbf{S}_r$. Hence, its area may be bounded uniformly from below, yielding a corresponding lower bound for $\mathcal{A}$, as in the proof of Theorem \ref{comp}.

\section{Monotonicity Proofs} 

\label{mp}

\subsection{The basic idea}

\label{bi}

In this section we will establish Theorems \ref{moncomp} and \ref{monnoncomp}. 
(Cf Fig. \ref{Fig2} or a vizualisation of either setting). As these results are restricted  to vacuum domains, a natural match would be to restrict to { \it (strictly) admissible vacuum extensions} $M_0$ of $\Om_0$ (defined by replacing the dominant energy condition by vacuum in point (4) of Definiton \ref{ae}). The basic idea is then to solve the vacuum Einstein Cauchy problem backward in time, 
and show that the data induced by the resulting spacetime  $D^-(M_0)$
on suitable hypersurfaces $M_t$ extending $\Om_t$ to the designated end  yield (strictly) admissible 
vacuum extensions. In the setting of Theorem 3, the basic result reads as follows.

\begin{prop}
\label{prop}
    Let $(\Omega_0,\mathrm{g}_0,\mathrm{k}_0)$ be a compact vacuum initial data set whose MOTS 
    boundary components are strictly stable, and let $(M_0,g_0,k_0)$ be a strictly admissible vacuum extension.
Let $\{(\Omega_t,\mathrm{g}_t,\mathrm{k}_t)\}_{t\in[-\varepsilon,0]}$ be a foliation by compact spacelike hypersurfaces of a spacelike MOTT domain contained within the past domain of dependence $D^-(\Om_0)$. \\
Then  there exists a positive $\varepsilon_1 <\varepsilon$ such that each  $(\Omega_t,\mathrm{g}_t,\mathrm{k}_t)$ has a strictly admissible vacuum extension $(M_t,g_t,k_t)$ for $t \in [-\varepsilon_1, 0]$.
\end{prop} 
\begin{proof}
Using local existence of the Cauchy problem for the vacuum Einstein equations we can 
 evolve a collar neighborhood $(V,\tilde{g}_0,\tilde{k}_0)$ of $\partial\Omega$ backwards in time to obtain a portion of a vacuum spacetime. Since the MOTT is spacelike, we may then smoothly extend $\Omega_t$ into spacetime to obtain an extension $(M_t,\tilde{g}_t,\tilde{k}_t)$, which first passes through the complement of the MOTT domain within $D^{-}(\Omega_0)$ and then through $D^{-}(V)$ to eventually coincide with the unbounded component of $(M_0\setminus V,\tilde{g}_0,\tilde{k}_0)$. This construction of $M_t$ is illustrated schematically in Figure~\ref{Fig3}. As these  extensions share a single asymptotically flat end they have the same  ADM mass. \\

\item[] {

    \begin{center}
    \includegraphics[width=0.8\textwidth]{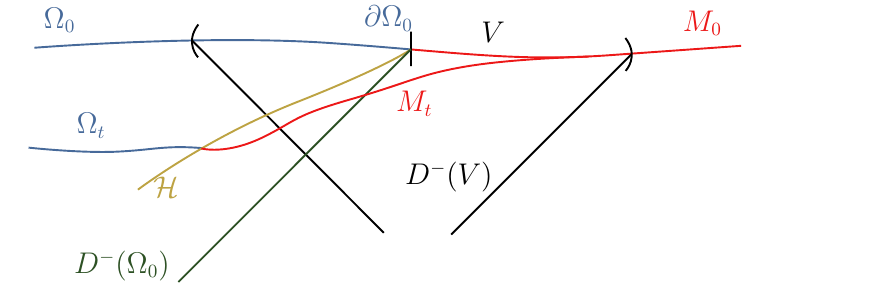}
        \captionof{figure}{Schematic view of the extension $M_t$.}
    \label{Fig3}
    \end{center}
\vspace{2ex}
}

 It is clear that each $M_t$ satisfies conditions (1)-(4) in  Definition \ref{ae}.  It remains to verify condition (5*) which will proceed by contradiction. Assume that there is no positive $\varepsilon_1$ such that (5*) holds for all $t\in[-\varepsilon_1,0]$. There is then a sequence of times $t_i\rightarrow 0$, such that $(M_{t_i},g_{t_i})$ contains a closed minimal surface that has nontrivial intersection with $M_{t_i} \setminus\mathrm{int}(\Omega_{t_i})$. It follows that the outermost minimal surface $\mathbf{\Sigma}_{t_i}\subset M_{t_i}$, which consists of stable embedded minimal 2-spheres \cite[Lemma 4.1 (i)]{HI}, also has nontrivial intersection with $M_{t_i} \setminus\mathrm{int}(\Omega_{t_i})$. Classical curvature estimates \cite[Theorem 3]{Schoen} for stable minimal hypersurfaces, combined with a uniform area bound arising from the outerminimizing property of outermost minimal surfaces as well as the uniformly controlled ambient geometry of $M_{t_i}$, imply that $\mathbf{\Sigma}_{t_i}$ subconverges smoothly to a closed stable minimal surface $\mathbf{\Sigma}_0 \subset M_0$. Moreover, this limit surface must have nontrivial intersection with $M_{0} \setminus\mathrm{int}(\Omega_{0})$, yielding a contradiction to strict admissibility of the extension $M_0$. We conclude that an $\varepsilon_1 >0$ exists for which $\{(M_t,g_t,k_t)\}_{t\in[-\varepsilon_1 ,0]}$ are strictly admissible extensions.
 \end{proof}

This seems close  to the required monotonicity proof in vacuum. However,  
the value $\varepsilon_1$ is potentially dependent on the chosen extension at time $t=0$
and might shrink to zero upon minimization over all smooth extensions.
This spoils the monotonicity argument.
While we were not able to fix this problem within the vacuum class, we
succeeded via  a substantial detour using `auxiliary'
Einstein-Vlasov matter, which we  sketch before giving the proofs.

\subsection{Proofs of Theorems 3 and 4}

Within the setting of Theorem 3, we first note that this result requires strict mean convexity of $\partial \Omega_0$, which is not required in Proposition \ref{prop}.
In fact the key step 5 of the following proof requires strict mean convexity for all $\partial \Omega_t$, which follows from
mean convexity of $\partial\Om_0$  by shrinking the initial interval $[-\varepsilon, 0)$ if necessary. As a lesson from the previous subsection, we avoid any further shrinking of the time interval in the following process  (steps 4-6) which consists of  modifying the extensions $(M_t,g_t,k_t)$ so that they become strictly admissible for all $t\in[-\varepsilon,0]$. 
This comes at the expense of leaving the vacuum class and modeling the given energy-momentum tensor by suitable Einstein-Vlasov particles which we do in step 2.
Proposition \ref{prop}, formulated for vacuum, carries over to Einstein-Vlasov  straightforwardly. 
However, the modeling itself requires a prior deformation of the given data to ones satisfying the strict dominant energy condition (step 1) while preserving strict admissibility. The latter property is shown by a compactness argument analogous to Proposition \ref{prop}.
The key step 5 then requires two further perturbations (steps 3 and 4). All deformations carried out below induce none or 
only small changes of the ADM mass which in particular do not affect the outcome of the final 
minimization  (step 7). 

\begin{proof} [Proof of Theorem 3] Let $(M_0,g_0,k_0)$ be a strictly admissible extension of the compact initial data set $(\Omega_0,\mathrm{g}_0,\mathrm{k}_0)$ whose boundary components are strictly stable and strictly mean convex MOTS. Let $\{(\Omega_t,\mathrm{g}_t,\mathrm{k}_t)\}_{t\in[-\varepsilon,0]}$ be a foliation by compact spacelike hypersurfaces of a spacelike MOTT domain contained within the past domain of dependence $D^-(\Om_0)$.  By shrinking $\varepsilon$ if necessary, we may assume without loss of generality that $\partial\Omega_t$ is strictly mean convex for $t\in[-\varepsilon,0]$. 
\begin{enumerate}
\item By Lemma \ref{tech} (Appendix A) there exists a perturbation of $(g_0,k_0)$ on $M_0 \setminus\Omega_0$ to strict dominant energy condition, while disturbing the ADM mass only slightly. Denote this new extension by $(M_0,\tilde{g}_0,\tilde{k}_0)$, and note that for sufficiently small perturbation parameter this is also a strictly  admissible extension. Indeed, only the strict admissibility condition (5*) requires some justification, and this is provided by the same type of compactness argument 
as in Proposition \ref{prop}. 

\item 
Now apply Lemma \ref{thm:EV-collar-interface} (Appendix B) to evolve a collar neighborhood $(V,\tilde{g}_0,\tilde{k}_0)$ of $\partial\Omega$ backwards in time to obtain a portion of a spacetime satisfying the dominant energy condition. Recalling the
construction of Proposition \ref{prop} (cf Fig. \ref{Fig3}),  we smoothly extend $\Omega_t$ into spacetime to obtain an extension $(M_t,\tilde{g}_t,\tilde{k}_t)$, which first passes through the complement of the MOTT domain within $D^{-}(\Omega_0)$ and then through $D^{-}(V)$ to eventually coincide with the unbounded component of $(M_0\setminus V,\tilde{g}_0,\tilde{k}_0)$. By abuse of notation going forward, we will denote $(M_t,\tilde{g}_t,\tilde{k}_t)$ by $(M_t,g_t,k_t)$.

Note that these extensions $(M_t,g_t,k_t)$ are strictly admissible, but only on the subinterval $[0,\varepsilon_1]$ by Proposition \ref{prop}.

\item We now perform another perturbation $\delta(t)$ with the goal of preserving the strict dominant energy condition of the data  $(M_t,g_{\delta(t)},k_{\delta(t)})$ on the whole interval
$[-\varepsilon,0]$.  We take $\delta(t)$ to be a smooth function on $[-\varepsilon,0]$ which vanishes at $t=0$ and is strictly positive for $t<0$; according to Lemma \ref{tech}, for each $t<0$ there exists a perturbation of the unbounded component $M_t^\infty$ of $M_t \setminus\Omega_t$ that yields a new asymptotically flat initial data set $(M_t,g_{\delta(t)},k_{\delta(t)})$ satisfying: the strict dominant energy condition holds at all points of $M_t^\infty$, and
\begin{equation}\label{mass}
|m_{ADM}(M_t,g_{\delta(t)},k_{\delta(t)})-m_{ADM}(M_t,g_t,k_t)|<\delta(t).
\end{equation}

\item A final perturbation will now turn a certain class of closed, stable minimal surfaces into strictly stable ones. This perturbation concerns $U_t$, the open subset of $M_t^\infty$ obtained by removing a large coordinate sphere in the asymptotic end and taking the bounded component. By applying White's bumpy metric theorem \cite[Theorem 3.1]{White}, we can find a nearby smooth metric $g'_{\delta(t)}$ which agrees with $g_{\delta(t)}$ outside of $U_t$, and has the property that any closed immersed stable minimal surface in $(M_t,g'_{\delta(t)})$ which intersects $U_t$ is strictly stable. Since the deformation carried out in step 4 guarantees that the data $(M_t,g_{\delta(t)},k_{\delta(t)})$ satisfy the strict dominant energy condition, so do  $(M_t,g'_{\delta(t)},k'_{\delta(t)})$ (where $k'_{\delta(t)}=k_{\delta(t)}$) when the present perturbation is sufficiently close to $g_{\delta(t)}$. Moreover, the mass is unchanged and therefore \eqref{mass} holds for this perturbation.

\item Recalling Definition \ref{ae} of  \textit{admissibility}
we now claim that the set of times $I\subset [-\varepsilon,0]$ such that $(M_t,g'_{\delta(t)},k'_{\delta(t)})$ is an admissible extension of $(\Omega_t,\mathrm{g}_t,\mathrm{k}_t)$, is the entire interval $[-\varepsilon,0]$. First, observe that conditions (1)-(4) are immediately satisfied, and thus it remains to verify  (5). The case $t=0$ is trivial since $g'_{\delta(0)}=g_0$ and $k'_{\delta(0)}=k_0$, so that $I$ is not empty. Next, notice that the arguments at the beginning of this section show that $I$ is open. More precisely, if this is not the case then there exists $t_{*}\in I$ and a sequence $t_j \rightarrow t_*$ with $t_j \notin I$ for each $j$. This yields a sequence of stable minimal surfaces, each element of which lies in $(M_{t_j},g'_{\delta(t_j)})$ and exits $\Omega_{t_j}$. These must then subconverge to a stable minimal surface of $(M_{t_*},g'_{\delta(t_*)})$ that lies within $\Omega_{t_*}$ because $t_{*}\in I$. On the other hand, since the $j$th member of the sequence of minimal surfaces intersects $M_{t_j}\setminus\Omega_{t_j}$, the limit minimal surface must have nontrivial intersection with $\partial\Omega_{t_*}$. However, this is impossible by the maximum principle as $\partial\Omega_{t_*}$ is strictly mean convex. To establish closedness, consider a sequence $t_j\in I$ with $t_j \rightarrow t_0 \geq-\varepsilon$. If $t_0 \notin I$, then there exists a closed minimal surface of $(M_{t_0},g'_{\delta(t_0)})$ that exits $\Omega_{t_0}$. It follows that the outermost minimal surface $\mathbf{\Sigma}_{t_0}$ must intersect $U_{t_0}$ as it cannot enter the part of the asymptotic end foliated by mean convex coordinate spheres. Thus, $\mathbf{\Sigma}_{t_0}$ is strictly stable. The implicit function theorem, based at $\mathbf{\Sigma}_{t_0}$, will then yield a closed minimal surface in $(M_{t_j},g'_{\delta(t_j)})$ for all $j$ sufficiently large, which has nontrivial intersection with $M_{t_j}\setminus\Omega_{t_j}$. This contradicts the assumption that $t_j \in I$, so that in fact $t_0 \in I$ and $I$ is closed, establishing the claim. 

\item As final step in the preparation  of the data we observe that the admissible extensions $\{(M_t,g'_{\delta(t)},k'_{\delta(t)})\}_{t\in[-\varepsilon,0]}$ constructed above are in fact strictly admissible. To see this, assume that one of them, $M_t$, is not strictly admissible. Then there exists a closed minimal surface in $M_{t}$ that leaves $\mathrm{int}(\Omega_t)$. However, by admissibility this minimal surface must be contained within $\Omega_t$, and hence it has a nontrivial intersection with $\partial\Omega_{t}$. This contradicts the maximum principle, and establishes the desired property.

\item To address monotonicity of the Bartnik mass, let $t_2\in[-\varepsilon,0]$, and choose a function $\delta_2(t)$ that is smooth on $[-\varepsilon,t_2]$ with the property that it vanishes at $t_2$ and is strictly positive for $t<t_2$. Then applying the above constructions with $t_2$ playing the role of `starting time', for any starting strictly admissible extension $(M_{t_2},g_{t_2},k_{t_2})$ we obtain strictly admissible extensions $(M_{t_1},g'_{\delta_2(t_1)},k'_{\delta_2(t_1)})$ for each $t_1\in[-\varepsilon,t_2]$ such that 
\begin{equation}\label{mass1}
|m_{ADM}(M_{t_1},g'_{\delta_2(t_1)},k'_{\delta_2(t_1)})-m_{ADM}(M_{t_2},g_{t_2},k_{t_2})|< \delta_2(t_2).
\end{equation}
Thus, we have shown that for any $t_1\in[-\varepsilon, t_2]$ there exists a strictly admissible extension $M_{t_1}$ having an asymptotic end whose ADM mass is arbitrarily close to that of $M_{t_2}$. It follows immediately that the Bartnik masses satisfy $m(\Omega_{t_1})\leq m(\Omega_{t_2})$.
\end{enumerate}

\end{proof}

\begin{proof}[Proof of Theorem 4]
  Theorem \ref{monnoncomp} can now be established in a similar manner with straightforward modifications.
  \end{proof}

\begin{remark}
This remark expands the discussion at the beginning of Sect.  \ref{mono} where we emphasized the difference between  temporal monotonicity on the one hand, and monotonicity of nested domains within a $t=const = 0$ slice $M_0$ on the other hand. However, these issues are related when we consider the limit that ${\mathcal H}$ 
approaches $M_0$ in the former setting. Then in particular Theorem \ref{monnoncomp} continues to hold and guarantees monotonicity of nested domains within $\Omega_0$ having boundaries sufficiently close to $\partial\Omega_0$, as  the requirement of strict mean convexity of $\partial \Om_0$ excludes adjacent minimal surfaces.
\end{remark}

\begin{remark} 
\label{jumps}
This final comment concerns a possible extension of Theorem \ref{moncomp}
 to the case where the outermost MOTT  $\mathcal{H} = \{\partial \Om\}$ 
is not smooth anymore. In a spacetime foliated by spacelike slices where some initial
slice $M_0$ contains a MOTS, the future propagation will generically 
reveal `jumps' of the outermost MOTT. A systematic investigation of this behavior 
has been carried out in  \cite{AMMS}. Regarding monotonicity of $m(\Omega)$ along such a jump,
it holds provided minimal surfaces are absent on $M_t$ (since any admissible extension of the target MOTS of the jump  yields an admissible extension of the origin). In the generic case, however, we expect the  monotonicity problem to be very intricate.
\end{remark}
~\\
\noindent\textbf{Acknowledgements.} The authors would like to thank Jan Metzger for discussions at the early stages of this work. MK acknowledges support from NSF Grant DMS-2405045.  MM acknowledges support under projects PID2024-158938NB-I00 (MICIU) and SA097P24 (JCyL).
The research of WS was funded  by the Austrian Science Fund (FWF) [DOI 10.55776/P35078]. 

\appendix
\section{Appendix}
\label{appA}

In the proof of Theorems \ref{moncomp} and \ref{monnoncomp}, a technical perturbation to strict dominant energy condition
is utilized, in which the ADM mass is also only slightly disturbed. It follows from a combination of existing perturbation results, and is recorded here for completeness.

\begin{lemma}\label{tech}
Let $(M,g,k)$ be an asymptotically flat initial data set with one end and compact boundary, satisfying the dominant energy condition. For any $\delta>0$, there exists a new asymptotically flat initial data set $(M,\tilde{g},\tilde{k})$ admitting the following properties. 
\begin{enumerate}
\item $(\tilde{g},\tilde{k})$ is $\delta$-close to $(g,k)$ in a weighted H\"{o}lder space involving at least two derivatives. In particular, $(\tilde{g},\tilde{k})$ agrees with $(g,k)$ up to second order at the boundary of $M$.

\item The strict dominant energy condition holds globally for $(\mathrm{int}(M),\tilde{g},\tilde{k})$.

\item The relation between the two ADM masses is given by
\begin{equation}\label{mass1}
|m_{ADM}(M,\tilde{g},\tilde{k})-m_{ADM}(M,g,k)|<\delta.
\end{equation}
\end{enumerate}
\end{lemma}

\begin{proof}
Choose a large coordinate sphere $S_r$ in the asymptotic end, and denote the bounded component of $M\setminus S_r$ by $M_r$. By 
\cite[Theorem 1.4]{CorvinoHuang} there exists a nearby new asymptotically flat initial data set $(M,g_r,k_r)$ which agrees with $(M,g,k)$ on $M_r$, and agrees with Kerr-Newman initial data on $M\setminus M_{2r}$. Moreover, the mass of the new data may be made as close to the original as desired by taking $r$ sufficiently large. Choosing the Kerr-Newman data to have nonvanishing electromagnetic fields then yields a strict dominant energy condition for the new data on $M\setminus M_{2r}$. Next, we deform $(M,g_r,k_r)$ by using a slight generalization of the statement in \cite[Theorem 8]{HuangLee}, to obtain the desired data set $(M,\tilde{g},\tilde{k})$. In particular, $(\tilde{g},\tilde{k})$ agrees with $(g_r,k_r)$ on $M\setminus M_{3r}$ and admits a strict dominant energy condition on $M_{3r}$. This is achieved by applying \cite[Theorem 8]{HuangLee} on the compact closure $\bar{M}_{3r}$ and choosing a small ball $B\subset M_{3r}\setminus M_{2r}$ to obtain: given $\varepsilon>0$ we may choose a sufficiently small positive function $u\in C^{0}_c(\mathrm{int}(M_{3r}))$ such that
\begin{equation}
\tilde{\mu}-|\tilde{J}|\geq (\mu_r-|J_r|) +u -\varepsilon\mathbf{1}_{B}\quad\quad\text{ on }M_{3r},
\end{equation}
where $\mathbf{1}_{B}$ is the indicator function for $B$, and $(\tilde{\mu},\tilde{J})$, $(\mu_r,J_r)$ are the energy-momentum density of matter fields for $(M,\tilde{g},\tilde{k})$, $(M,g_r,k_r)$ respectively. Since $\mu_r -|J_r|\geq c>0$ for some constant $c$ on $M_{3r}\setminus M_{2r}$, by taking $\varepsilon$ small enough it follows that $\tilde{\mu}-|\tilde{J}|>0$ on $M$. We note that the published statement of \cite[Theorem 8]{HuangLee} allows for a positive $u$ only on any $V\Subset M_{3r}$. However, this is based on \cite[Theorem 3.3]{HuangLee}, whose proof in turn is essentially the same as 
\cite[Theorem 3.1]{CorvinoHuang}. 
This latter result is stated using certain weighted H\"{o}lder spaces instead of compactly contained domains. Therefore, a slight generalization of \cite[Theorem 8]{HuangLee} utilizing the weighted H\"{o}lder spaces is also valid, and this is what we have employed here. In fact, Huang-Lee make reference to such a generalization for \cite[Theorem 3.3]{HuangLee} in the first paragraph of its proof.
\end{proof}

\section{Appendix}
\label{appB}

In this section we will show how to realize certain types of initial data in the context of
Vlasov matter, with the purpose of embedding portions of initial data satisfying the dominant
energy condition into a Lorentzian manifold satisfying the spacetime dominant energy condition.
It should be noted that Gl\"{o}ckle \cite[Theorem 4]{Glockle} has shown that there exist smooth DEC initial
data sets which do not arise as the induced metric and second fundamental from for a spacelike slice
within a smooth spacetime satisfying the \textit{spacetime dominant energy condition}, namely that $T(X,Y)\geq 0$
for any two future causal vectors $X$, $Y$ where $T$ is the stress-energy tensor. We begin with two preliminary propositions.

\begin{prop}[Strict DEC realization by smooth Vlasov data]\label{lem:strict-DEC-realization}
Let $(M,g,k)$ be an initial data set satisfying the strict dominant energy condition $\mu>|J|$ on a
compact subset $\mathcal{K}\subset M$. Then there exists a nonnegative (Vlasov distribution) function
$f_0\in C^\infty(TM|_\mathcal{K})$
such that
\begin{enumerate}
\item $f_0(x,\cdot)$ is compactly supported in each fiber $T_x\mathcal{K}$;
\item for every $x\in \mathcal{K}$ the following identities hold
\begin{equation}
\int_{T_xM} \!\!\!f_0(x,v)\,\sqrt{1+|v|_g^2}\,dv_{g_x}=\mu(x),\quad
\int_{T_xM} \!\!\!f_0(x,v)\,v\,dv_{g_x}=J(x).
\end{equation}
\end{enumerate}
\end{prop}

\begin{proof}
Since $\mathcal{K}$ is compact we find $\sigma:=\min_{\mathcal{K}}(\mu-|J|)>0$.
Choose a nonnegative smooth function $\chi\in C_c^\infty([0,\infty))$,
supported in $[0,r_0^2]$ for some $r_0>0$, and normalize it so that if
$\eta(x,v):=\chi(|v|_g^2)$ then
\begin{equation}
\int_{T_xM}\eta(x,v)\,dv_{g_x}=1
\qquad\text{for all }x\in \mathcal{K}.
\end{equation}
Because $\eta(x,\cdot)$ is radial in each fiber, one also has
\begin{equation}
\int_{T_xM}\eta(x,v)\,v\,dv_{g_x}=0
\qquad\text{for all }x\in \mathcal{K}.
\end{equation}
Next, for $q\in T_xM$, define the translated bump
$\eta_q(x,v):=\eta(x,v-q)$
and observe that
\begin{equation}
\int_{T_xM}\eta_q(x,v)\,dv_{g_x}=1,
\qquad
\int_{T_xM}\eta_q(x,v)\,v\,dv_{g_x}=q.
\end{equation}
Since $\eta_q(x,\cdot)$ is supported in the $g_x$-ball $B_{r_0}(q)\subset T_xM$, we have
\begin{equation}
\sqrt{1+|v|_g^2}\le |q|_g+\sqrt{1+R^2}
\qquad\text{on }\operatorname{supp}\eta_q(x,\cdot),
\end{equation}
and therefore
\begin{equation}
\mathcal{E}(x,q):=\int_{T_xM}\eta_q(x,v)\,\sqrt{1+|v|_g^2}\,dv_{g_x}\le |q|_g+\sqrt{1+r_0^2}.
\end{equation}
By choosing $c>0$ so large that $c^{-1}\sqrt{1+r_0^2}<\sigma$, it follows that
for every $x\in \mathcal{K}$ we have
\begin{equation}
\frac1c\,\mathcal{E}(x,cJ(x))
\le |J(x)|+c^{-1}\sqrt{1+r_0^2}
< |J(x)|+\sigma
\le \mu(x).
\end{equation}
Hence the function
\begin{equation}
a(x):=\frac{\mu(x)-c^{-1}\mathcal{E}(x,cJ(x))}{e_0},\quad\quad e_0:=\int_{T_xM}\eta(x,v)\,\sqrt{1+|v|_g^2}\,dv_{g_x},
\end{equation}
is smooth and nonnegative on $\mathcal{K}$; note that $e_0$ is independent of $x$.
We may now set
\begin{equation}
f_0(x,v):=a(x)\,\eta(x,v)+\frac1c\,\eta_{cJ(x)}(x,v).
\end{equation}
Then $f_0\ge 0$, and for each $x\in \mathcal{K}$ the function $f_0(x,\cdot)$ is compactly
supported in $T_xM$. To compute the moments, observe that
\begin{equation}
\int_{T_xM} \!\!\!f_0(x,v)\,v\,dv_{g_x}
=
a(x)\int_{T_xM}\!\!\!\eta(x,v)\,v\,dv_{g_x}
+\frac1c\int_{T_xM}\!\!\!\eta_{cJ(x)}(x,v)\,v\,dv_{g_x},
\end{equation}
so that according to the identities above
\begin{equation}
\int_{T_xM} f_0(x,v)\,v\,dv_{g_x}
=
0+\frac1c\,(cJ(x))
=
J(x).
\end{equation}
Moreover
\begin{equation}
\int_{T_xM} f_0(x,v)\,\sqrt{1+|v|_g^2}\,dv_{g_x}
=
a(x)e_0+\frac1c\,\mathcal{E}(x,cJ(x))
=
\mu(x).
\end{equation}
Thus $f_0$ is a smooth nonnegative Vlasov datum on $TM|_\mathcal{K}$, with the prescribed
energy and momentum moments.
\end{proof}

\begin{prop}[Vlasov realization up to a flat vacuum boundary]\label{lem:flat-boundary-vlasov}
Let $(M,g,k)$ be an initial data set, and let $\mathcal{K}\subset M$ be a compact domain with smooth boundary. Assume that
$\mu>|J|$ on $\mathrm{int}(\mathcal{K})$, and that both $\mu$ and $J$ vanish to infinite order along a collection of boundary components $\partial' \mathcal{K}\subset\partial\mathcal{K}$.
Then there exists a nonnegative function $f_0\in C^\infty\!\big(TM|_{\mathcal{K}}\big)$
such that
\begin{enumerate}
\item $f_0$ is compactly supported in each fiber $T_xM$;
\item $f_0$ vanishes to infinite order along $T(\partial' \mathcal{K})$, hence extends by zero to a smooth nonnegative function on $TM$;
\item for every $x\in \mathcal{K}$ the following identities hold
\begin{equation}
\int_{T_xM} \!\!\!f_0(x,v)\,\sqrt{1+|v|_g^2}\,dv_{g_x}=\mu(x),\quad
\int_{T_xM} \!\!\!f_0(x,v)\,v\,dv_{g_x}=J(x).
\end{equation}
\end{enumerate}

\end{prop}

\begin{proof}
Choose a smooth boundary defining function
\begin{equation}
s\in C^\infty(\mathcal{K}), \qquad s\ge 0,\qquad \partial' \mathcal{K}=\{s=0\},\qquad \mathrm{int}(\mathcal{K})=\{s>0\}.
\end{equation}
Let $\{\chi_n\}_{n\ge 1}$ be a dyadic partition of unity so that
\begin{align}\label{b15}
\begin{split}
\sum_{n=1}^\infty &\chi_n =1 \quad \text{on } \mathrm{int}(\mathcal K), \\
\operatorname{supp}\chi_n \subset& \{2^{-n-2-n_0}<s<2^{-n-n_0}\} \quad \text{ for } n\geq 2,
\end{split}
\end{align}
and with $\chi_1$ supported away from $\partial' \mathcal K$; here $n_0$ is fixed sufficiently large so that the annular regions are smooth and contained within $\mathcal{K}$.
For each $n\ge 0$ let $\mathcal K_n\Subset \mathrm{int}(\mathcal K)$ be the closure of the annular region appearing
in \eqref{b15}, so that $\operatorname{supp}\chi_n \subset \mathcal K_n$. Since $\mu>|J|$ on $\mathrm{int}(\mathcal K)$, for each $n$ there is a positive constant $\sigma_n:=\min_{\mathcal K_n}(\mu-|J|)>0$.
Hence, on each fixed compact set $\mathcal K_n$, the strict dominant energy condition holds with a uniform positive gap. By the strict-DEC realization Proposition \ref{lem:strict-DEC-realization}, there exist nonnegative smooth functions $h_n\in C^\infty\!\big(TM|_{\mathcal K_n}\big)$
which are compactly supported in each fiber and satisfy
\begin{equation}
\int_{T_xM} h_n(x,v)\,\sqrt{1+|v|_g^2}\,dv_{g_x}=\mu(x),
\qquad
\int_{T_xM} h_n(x,v)\,v\,dv_{g_x}=J(x),
\end{equation}
for all $x\in \mathcal K_n$. 

Let $\pi:TM\to M$ be the bundle projection, and set
$f_n:=(\chi_n\circ \pi)\,h_n$.
Then $f_n\ge 0$, $f_n$ is smooth on $TM|_{\mathcal K_n}$, and by linearity of the moment map
\begin{equation}
\int_{T_xM} f_n(x,v)\,\sqrt{1+|v|_g^2}\,dv_{g_x}=\chi_n(x)\mu(x),
\end{equation}
\begin{equation}
\int_{T_xM} f_n(x,v)\,v\,dv_{g_x}=\chi_n(x)J(x).
\end{equation}
Now define
\begin{equation}
f_0:=\sum_{n=1}^\infty f_n
\qquad \text{on } TM|_{\mathrm{int}(\mathcal K)}.
\end{equation}
Because the partition is locally finite on $\mathrm{int}(\mathcal K)$, this is pointwise well-defined and smooth in the interior. Furthermore, its moments are given by
\begin{equation}
\sum_{n=1}^\infty \chi_n \mu=\mu,
\qquad
\sum_{n=1}^\infty \chi_n J=J,
\end{equation}
hence $f_0$ realizes exactly the prescribed pair $(\mu,J)$ on $\mathrm{int}(\mathcal K)$.

It remains to establish smooth extension by zero across $\partial' \mathcal K$. Since $\mu$ and $J$ vanish to infinite order along $\partial' \mathcal K$, for every $N\ge 1$ and every integer $l\ge 0$ there is a constant $C_{N,l}$ such that on the collar $\{2^{-n-2-n_0}<s<2^{-n-n_0}\}$ we have
\begin{equation}
\|\mu\|_{C^l}+\|J\|_{C^l}\le C_{N,l}\,2^{-nN}.
\end{equation}
The strict-DEC realization on each fixed compact annulus depends smoothly on the data, so that the corresponding $h_n$ satisfy analogous tame estimates
\begin{equation}
\|h_n\|_{C^l}\le C'_{N,l}\,2^{-nN}.
\end{equation}
Since the derivatives of $\chi_n$ grow at most polynomially in $2^n$, the series
$\sum_{n=1}^\infty f_n$ converges in every $C^l$-norm up to the boundary $T(\partial'\mathcal K)$, and the limit vanishes to infinite order there.
\end{proof}

We will now evolve the Einstein-Vlasov system near a vacuum/nonvacuum interface, in order to
realize certain types of DEC initial data as spacelike hypersurfaces in DEC spacetimes.
Although Gl\"{o}ckle \cite[Theorem 4]{Glockle} has found examples where this is not possible,
they require conditions that are far from the setting of our next result. In particular, the
counterexamples need an open subset on which $\mu=|J|$ together with a lack of $C^2$ regularity for the map
\begin{equation}
x \mapsto
\begin{cases}
    \mu^{-1}(x) (J(x)\otimes J(x)) & \mu(x) \neq 0 ,\\
    0   & \mu(x) = 0.
\end{cases}
\end{equation}

\begin{lemma}\label{thm:EV-collar-interface}
Let $(M,g,k)$ be a bounded initial data set, and let $\Omega\subset M$ be a compact domain with smooth boundary. Assume that the data are vacuum ($\mu=|J|=0)$ on $\Omega$, and satisfy the strict dominant energy condition $\mu>|J|$ on $M\setminus\Omega$. Then there exists a short time smooth Einstein-Vlasov development realizing the given initial data on $M$ as a spacelike hypersurface, and satisfying the spacetime dominant energy condition. 

\end{lemma}

\begin{proof}
Consider the compact domain $\mathcal K:=\overline{M}\setminus \operatorname{Int}(\Omega)$.
Since the strict dominant energy condition holds on $\mathrm{int}(\mathcal K)=M\setminus \Omega$, and $\mu$ and $J$ vanish to infinite order on $\partial \Omega=:\partial'\mathcal{K}$, we may apply Proposition \ref{lem:flat-boundary-vlasov} to find a smooth nonnegative function $\widetilde f_0\in C^\infty\!\big(TM|_{\mathcal K}\big)$
which is compactly supported in each fiber, vanishes to infinite order along $T(\partial\Omega)$, and satisfies
$\mu_{\widetilde f_0}=\mu$ and $J_{\widetilde f_0}=J$ on $\mathcal K$.
Next, extend $\widetilde f_0$ by zero over $T\Omega$ to obtain a smooth nonnegative function
$f_0\in C^\infty(TM)$ such that $f_0=0$ on $T\Omega$.
Moreover, on $\Omega$ one has $\mu_{f_0}=\mu=0$ and $J_{f_0}=J=0$. Hence, the energy-momentum densities arising from the
Vlasov distribution function, and initial data, agree on all of $M$.
We then have that the triple $(M,g,k,f_0)$
satisfies the Einstein-Vlasov constraint equations, and by the standard local Cauchy theory for the Einstein-Vlasov system, there exists a short time smooth spacetime development \cite[Theorem 1]{And}
realizing the given initial data on $M$ as a spacelike hypersurface. Moreover, since $f_0 \geq 0$ the resulting spacetime satisfies the dominant energy condition.
\end{proof}

\end{document}